\documentclass[cmbright]{my_mmaauth}

\usepackage{amsmath}
\usepackage{moreverb}

\usepackage{caption}
\usepackage{subfig}

\usepackage[dvips,colorlinks,bookmarksopen,bookmarksnumbered,citecolor=red,urlcolor=red]{hyperref}

\usepackage{cite}

\newtheorem{definition}{Definition}[section]
\newtheorem{lemma}{Lemma}[section]
\newtheorem{proposition}{Proposition}[section]
\newtheorem{theorem}{Theorem}[section]

\newenvironment{proof}[1][Proof]{\begin{trivlist}
\item[\hskip \labelsep {\bfseries #1}]}{\end{trivlist}}

\begin{document}

\runninghead{A. Rachah, D. F. M. Torres}

\title{Predicting and controlling the Ebola infection}

\author{Amira Rachah\affil{$\!\,^{\dag}$}
and Delfim F. M. Torres\affil{$\!\,^{\ddag}$}\corrauth}

\address{\affilnum{$^{\dag}$}
Universit\'e Paul Sabatier, Institut de Math\'ematiques,
31062 Toulouse Cedex 9, France.
\affilnum{$^{\ddag}$}Center for Research and Development
in Mathematics and Applications (CIDMA), Department of Mathematics,
University of Aveiro, 3810-193 Aveiro, Portugal.}

\corraddr{Department of Mathematics,
University of Aveiro,
3810-193 Aveiro, Portugal.
Email: {\tt delfim@ua.pt}}


\begin{abstract}
We present a comparison between two different mathematical
models used in the description of the Ebola virus propagation
currently occurring in West Africa. In order to improve the
prediction and the control of the propagation of the virus,
numerical simulations and optimal control of the two
models for Ebola are investigated. In particular, we study
when the two models generate similar results.
\end{abstract}


\MOS{92D30; 93A30}


\maketitle

\vspace{-6pt}


\section{Introduction}

Ebola virus is a severe infection affecting currently several African countries,
mainly Guinea, Sierra Leone, and Liberia. It was first discovered in 1976 in the
Democratic Republic of Congo, near the Ebola river, where the virus takes its
name, but recently was also identified in West Africa \cite{barraya,joseph}.
The latest outbreak of the virus was not only the largest, it is also the
deadliest infectious disease outbreak ever known.
According to the World Health Organisation (WHO), ``it is thought that fruit
bats of the Pteropodidae family are natural Ebola virus hosts''.
It then spreads through human-to-human transmission,
becoming the deadliest pathogen. The early signs and
symptoms of the virus include a sudden onset of fever, intense weakness
and headache. Over time, symptoms become increasingly severe and include
diarrhea, raised rash, internal and external bleeding, from nose, mouth,
eyes and anus. As the virus spreads through the body, it damages the immune
system and organs \cite{legrand,anon1,peter,tara1,alton}. Ebola virus
is transmitted to an initial human by contact with an infected animal's
body fluid. On the other hand, human-to-human transmission can take place
with direct contact (through broken skin or mucous membranes in, for example,
the eyes, nose, or mouth) with blood or body fluids of a person who is sick with
or has died from Ebola. It is also transmitted indirectly via exposure
to objects or environment contaminated with infected secretions
\cite{borio,dowel,edward,tara2,anon2,okwar}. There is yet no licensed
treatment proven to neutralise the virus, but a range of blood, immunological
and drug therapies are under development.

Mathematical models are a powerful tool for investigating human infectious
diseases, such as Ebola virus, contributing to the understanding of the dynamics
of disease and providing useful predictions about the potential transmission of
a disease and the effectiveness of possible control measures, which can provide
valuable information for public health policy makers \cite{diekman,delf,MR2719552}.
Cases of success include influenza and small pox, for which epidemic models have
provided the foundation for the best vaccination practices \cite{longini,kretz}.
Epidemic models date back to the early twentieth century, to a set of three
articles from 1927, 1932, and 1933 by Kermack and McKendrick, whose models were
used for modelling the plague and cholera epidemics
\cite{Kermack:McKendrick:1,Kermack:McKendrick:2,Kermack:McKendrick:3}.
The most commonly implemented models in epidemiology are the SIR and SEIR models.
The SIR model consists of three compartments: Susceptible individuals $S$,
Infectious individuals $I$, and Recovered individuals $R$. In many infectious
diseases there is an exposed period after the transmission of the infection
from susceptible to potentially infective members but before these potential
infective can transmit infection. Then an extra compartment is introduced,
the so called exposed class $E$, and we use compartments $S$, $E$, $I$ and $R$
to give a generalization of the basic SIR model \cite{herbert,brauer}.

When analyzing a new outbreak, the researchers usually start with the SIR
and SEIR models to fit the available outbreak data, obtaining estimates
for the parameters of the model \cite{brauer}. In the modelling of the spreading
mechanism of the Ebola virus currently affecting several African countries,
Rachah and Torres used the SIR \cite{MyID:321} and the SEIR \cite{symcomp}
basic models. In their work, they used parameters identified from recent data
of the WHO to describe the behavior of the virus.
After modelling and simulation steps, they applied optimal control techniques
in order to understand how the spread of the virus may be controlled, e.g.,
through education campaigns, immunization or isolation
\cite{gaff,vacc_opt2,vacc_opt4,vacc_opt3,vacc_opt5,vacc_opt1}.

Here we discuss the differences between the SIR and SEIR models,
previously presented by Rachah and Torres in their study of the
description of the behavior of the Ebola virus \cite{MyID:321,symcomp}.
For a non-integer order (fractional) SEIR model see \cite{Area}.
Our comparison is based on the numerical resolution of the models.
In order to study which model can better control the propagation of the virus
into population, we discuss the comparison between the SIR and SEIR models
by studying the optimal control of the Ebola virus. The manuscript is organized
as follows. In Section~\ref{Sec:2} we briefly recall the SIR and SEIR models
for Ebola infection. Their numerical simulation is carried out in
Section~\ref{Sec:3}, where we emphasise the differences between them.
The best strategies for the control of the propagation of Ebola into
the populations, accordingly to both models, are investigated
in Section~\ref{Sec:4}. We end with Section~\ref{Sec:conc} of conclusion.


\section{The SIR and SEIR models for Ebola infection}
\label{Sec:2}

SIR and SEIR models provide the foundations of mathematical
modelling in epidemiology, being widely used in practice,
to describe real world problems in epidemiology.
In this section, we briefly describe the two mathematical models.

In the modelling of Ebola virus, Rachah and Torres \cite{MyID:321}
used the basic SIR (Susceptible--Infectious--Recovery) model.
In the description of the transmission of Ebola virus by the SIR
model, the population is divided into three groups:
\begin{itemize}
\item the Susceptible compartment $S(t)$ denotes individuals who are susceptible
to catch the virus, and so might become infectious if exposed;

\item the Infectious compartment $I(t)$ denotes infectious individuals who are
suffering the symptoms of Ebola and are able to spread  the virus through
contact with the susceptible class of individuals;

\item the ``Recovered'' compartment $R(t)$ denotes individuals that do not
affect the transmission dynamics in any way.
\end{itemize}
The model is described by the following system of nonlinear
ordinary differential equations (ODEs):
\begin{equation}
\label{SIRmodel}
\begin{cases}
\dfrac{dS(t)}{dt} = -\beta S(t)I(t),\\[0.3cm]
\dfrac{dI(t)}{dt} =  \beta S(t)I(t) - \mu I(t),\\[0.3cm]
\dfrac{dR(t)}{dt} =  \mu I(t),
\end{cases}
\end{equation}
where $\beta >0$ is the infection rate and $\mu >0$ is the recovered rate.
The  initial conditions are given by
 \begin{equation}
\label{eq4:SIR}
S(0)=S_0>0,
\quad I(0)=I_0>0,
\quad R(0)=0.
\end{equation}

The SEIR model is an extension of the basic SIR model, where an extra compartment
is introduced, the so called exposed compartment $\hat{E}(t)$. The Exposed
compartment $\hat{E}(t)$ denotes the individuals who are infected but the
symptoms of the virus are not yet visible. The transmission of the virus
is then described by the following set of nonlinear ODEs:
\begin{equation}
\label{eq1:SEIR}
\begin{cases}
\dfrac{d\hat{S}(t)}{dt} = -\hat{\beta} \hat{S}(t)\hat{I}(t),\\[0.3cm]
\dfrac{d\hat{E}(t)}{dt}
= \hat{\beta} \hat{S}(t)\hat{I}(t) - \hat{\gamma} \hat{E}(t),\\[0.3cm]
\dfrac{d\hat{I}(t)}{dt} = \hat{\gamma} \hat{E}(t) - \hat{\mu} \hat{I}(t),\\[0.3cm]
\dfrac{d\hat{R}(t)}{dt} = \hat{\mu} \hat{I}(t),
\end{cases}
\end{equation}
where $\hat{\beta} >0$ is the contact rate (transmission rate),
$\hat{\gamma}$ is the infection rate, and $\hat{\mu} >0$ is the recovered rate.
The initial conditions are given by
 \begin{equation}
\label{eq5:SEIR}
\hat{S}(0)=\hat{S}_0>0,
\quad \hat{E}(0)=\hat{E}_0 \geq 0,
\quad \hat{I}(0)=\hat{I}_0>0,
\quad \hat{R}(0)=0.
\end{equation}
The particularity of the SEIR model is in the exposed compartment,
which is characterized by infected individuals that cannot communicate the virus.
These individuals are in the so called latent period \cite{brauer}. For Ebola
virus, such stage makes all sense since it takes a certain time for an infected
individual to become infectious. During this period of time such individuals
are in the exposed compartment $\hat{E}(t)$, before they become infectious.
The infectious compartment at time $t$, denoted by $\hat{I}(t)$,
represents the individuals that are infected by the virus and are able to spread
it through contact with the susceptible class of individuals.
Finally, we have the ``recovered'' compartment $\hat{R}(t)$, which denotes
the individuals that do not affect the transmission dynamics in any way.
Figure~\ref{fig1:SEIR} shows the diagrammatic representation of the
virus progress in an individual.
\begin{figure}
\centering
\includegraphics[scale=0.5]{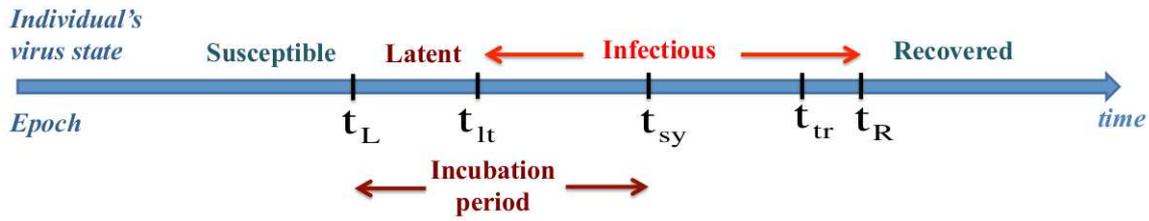}
\caption{Ebola virus progress in an individual
by using the SEIR model, where
$t_L$ is the time when infectious occurs,
$t_{lt}$ denotes latency to infectious transition,
$t_{sy}$ the time when symptoms appear,
$t_{tr}$ the time of first transmission to another susceptible,
and $t_R$ the time where the individual is no longer infectious.
\label{fig1:SEIR}}
\end{figure}


\subsection{Analysis of the SIR model}

In this subsection we study the properties of
the SIR model \eqref{SIRmodel}--\eqref{eq4:SIR}.
We begin with a trivial observation.

\begin{proposition}
\label{prop:N:SIR:const}
The total population $N$ being modelled by the SIR model \eqref{SIRmodel}
is constant in time, that is, $S(t) + I(t) + R(t) =  N$ for all $t \geq 0$.
\end{proposition}

\begin{proof}
Follows by adding the right-hand side of equations
\eqref{SIRmodel}: $S'(t)+I'(t) + R'(t) =0$.
\end{proof}

The following lemmas provide essential properties
of the SIR model \eqref{SIRmodel}--\eqref{eq4:SIR}.

\begin{lemma}
\label{lemma1:SIR}
Let $(S(t),I(t),R(t))$ be a solution of the SIR model
\eqref{SIRmodel}--\eqref{eq4:SIR}. Then $S(t) \geq 0$
and $I(t) \geq 0$ for all $t \geq 0$.
\end{lemma}

\begin{proof}
The proof is done by contradiction. We have that $S(t)$ and $I(t)$ are
continuous functions, $S(0)=S_0>0$ and $I(0)=I_0>0$.
Suppose that $S(t')<0$. Then, by the intermediate value theorem,
there exist $t_1 \in [0,t']$ such that $S(t_1)=0$. By using the first equation
of \eqref{SIRmodel}, we obtain that $S(t)=S(t_1)e^{f}=0$ for $t \geq t_1$,
where $f$ is the primitive of $-\beta I(t)$. Thus $S(t)=0$ for $t \geq t_1$,
which is a contradiction. The same arguments can be applied to $I(t)$.
\end{proof}

\begin{lemma}
\label{lemma2:SIR}
Let $(S(t),I(t),R(t))$ be a solution of the SIR model
\eqref{SIRmodel}--\eqref{eq4:SIR}. Then $S(t)+I(t) \leq N$
and $0 \leq R(t) \leq N$ for all $t\geq 0$.
\end{lemma}

\begin{proof}
The sum of the first two equations of \eqref{SIRmodel} give us
$\dfrac{d}{dt} \left[ S(t) + I(t) \right] = -\mu I(t) \leq 0$.
Now, by using Lemma~\ref{lemma1:SIR} and $S(0)+I(0)=N$, we obtain
that $S(t)+I(t) \leq N$.
From Proposition~\ref{prop:N:SIR:const} we have $N-S(t)-I(t)=R(t)$.
We conclude that $R(t) \geq 0$. From Lemma~\ref{lemma1:SIR},
we have $S(t) + I(t) \geq 0$, which implies that $N - S(t) - I(t) \leq N$.
Therefore, $R(t) = N - S(t) - I(t) \leq N$.
\end{proof}

In what follows we use $p$ to denote
the relative removal rate parameter.

\begin{definition}
Given a SIR model \eqref{SIRmodel}, we introduce the quantity
$p:= \dfrac{\mu}{\beta}$,
called the relative removal rate parameter.
\end{definition}

\begin{lemma}
\label{lemm3:SIR}
Let $(S(t),I(t),R(t))$ be a solution of the SIR model
\eqref{SIRmodel}--\eqref{eq4:SIR}. Then,
$I(S(t)) =  I_0  +  S_0 - S(t) + p\ln\left(\dfrac{S(t)}{S_0}\right)$.
\end{lemma}

\begin{proof}
From \eqref{SIRmodel}--\eqref{eq4:SIR} one has
\begin{equation}
\label{eq9:SIR}
\frac{dI}{dS} = \dfrac{\beta S(t) - \mu}{-\beta S(t)}
= \dfrac{\mu}{\beta S(t)} - 1,
\end{equation}
By integrating
$dI = \left[ \dfrac{\mu}{\beta S(t)}   - 1 \right]dS$,
we obtain that
$$
\int_{0}^{t}dI = \int_{0}^{t} \left[ \dfrac{\mu}{\beta S(t)}   - 1 \right]dS,
$$
so that
$I(t)-I_0=\dfrac{\mu}{\beta }\ln(S(t))-\dfrac{\mu}{\beta }\ln(S_0)-S(t)-S_0$.
Therefore,
$I(S(t)) =  I_0  +  S_0 - S(t) + p\ln\left(\dfrac{S(t)}{S_0}\right)$,
where $p=\dfrac{\mu}{\beta }$.
\end{proof}

From Lemma~\ref{lemm3:SIR}, we know that the quantity
$\dfrac{\mu}{\beta} S - 1 = p S -1$ is positive
if $S < p$ and negative if $S > p$. Hence, $I(S)$
is an increasing function of $S$ for $S < p$ and
a decreasing function of $S$ for $S > p$. From \eqref{eq9:SIR},
we see  that $I(0)=-\infty$ and $I(S_0)=I(0)>0$. Therefore,
there exists a unique point $S_\infty$, $0 < S_\infty < S_0$,
such that $I(S_\infty) = 0$  and  $I(S) > 0$ for $  S_\infty < S \leq S_0$.
The point $(S_\infty,0)$ is an equilibrium point of the first two equations
of \eqref{SIRmodel} since both $\dfrac{dS}{dt}$ and $\dfrac{dI}{dt}$
vanish when $I=0$.

\begin{lemma}
\label{lemma4:SIR}
Let $(S(t),I(t),R(t))$ be a solution of the SIR model
\eqref{SIRmodel}--\eqref{eq4:SIR} in
$T = \left\{ (S,I)  :  S \geq 0,  I \geq 0,   S+I \leq N \right\}$.
Then $0 < S(t) \leq S_0$ and
\begin{equation}
\label{eq6:SIR}
S(t)=S_0 e^{\dfrac{\beta\left(R(t)-R(0)\right)}{\mu}}
\geq S_0 e^{-\frac{\beta N}{\mu}}
\quad  \text{for all} \quad t > 0.
\end{equation}
\end{lemma}

\begin{proof}
By dividing the first and third equations of \eqref{SIRmodel},
we obtain that $\dfrac{dS}{dR} = -\dfrac{\beta S}{\mu}$. Thus
$\displaystyle \int \frac{dS}{S} = \int -\frac{\beta}{\mu} dR$.
By \eqref{eq4:SIR}, we know that $\ln\left(\frac{S(t)}{S(0)}\right)
= -\frac{\beta}{\mu}R(t) + \frac{\beta}{\mu}R(0)$,
so that $\ln\left(S(t)\right) = \ln\left(S(0)\right)
+ \frac{\beta\left(R_0 - R(t)\right)}{\mu}$.
We conclude that $S(t) =  S_0e^{-\frac{\beta\left(R(t) - R_0\right)}{\mu}}$.
From Lemma~\ref{lemma2:SIR}, $0\leq R(t) \leq N$ and we obtain that
$S_0 e^{-\frac{\beta N}{\mu}} \leq S_0
e^{-\frac{\beta\left(R(t) - R_0\right)}{\mu}} \leq S_0$.
Because $S_0>0$, we conclude that $0 < S(t) \leq S_0$ for all $t \geq 0$.
\end{proof}

The first conclusion we get from Lemma~\ref{lemma4:SIR}
is that $S$ is always a positive value,
hence there always remains some susceptible who are never infected.
We can also compute the number of susceptible at any time $t$
by using \eqref{eq6:SIR}. The second conclusion is:
if a small group of infectious is inserted into a group of susceptible
$S_0$ and $S_0 < p$, then the virus will die out rapidly. On the other hand,
if $S_0 > p$, then $I(t)$ increases as $S(t)$ decreases to $p$,
where it achieves its maximum value.

\begin{lemma}
\label{lemm5:SIR}
Let $(S(t),I(t),R(t))$ be a solution of the SIR model
\eqref{SIRmodel}--\eqref{eq4:SIR}. Then $S(t) \to S_\infty$
as $t \to \infty$ and  $R(t) \to  R_\infty$ as $t \to \infty$,
where $S_\infty$ and $R_\infty$ are finite numbers.
\end{lemma}

\begin{proof}
From the first equation of \eqref{SIRmodel}
and Lemma~\ref{lemma1:SIR}, we have
$\frac{dS(t)}{dt} = -\beta S(t)I(t) \leq 0$
so that $S(t)$ is a decreasing function and
$\lim_{t\to \infty}S(t) = S_\infty$,
where $S_\infty$ is a finite number. From the third
equation of \eqref{SIRmodel}, we have
$\frac{dR(t)}{dt} = \mu I(t) \geq 0$
so that $R(t)$ is an increasing function. By Lemma~\ref{lemma2:SIR},
$\lim_{t\to \infty}R(t) = R_\infty$, where $R_\infty$ is a finite number.
\end{proof}

\begin{lemma}
\label{lemm6:SIR}
Let $(S(t),I(t),R(t))$ be a solution of the SIR model
\eqref{SIRmodel}--\eqref{eq4:SIR}. Then
$I(t) \to 0$ as $t \to \infty$.
\end{lemma}

\begin{proof}
By \eqref{SIRmodel}--\eqref{eq4:SIR},
and knowing that from Lemma~\ref{lemm5:SIR}
$\lim_{t \to \infty} R(t) = R_\infty$, so that
$\lim_{t \to \infty} \mu \int_0^t I(s) ds= \dfrac{R_\infty}{\mu}$,
we obtain
\begin{equation*}
\int\dfrac{dR}{dt} = \int \mu I dt
\Rightarrow R(t) = \mu \int_{0}^{t} I(s) ds
\Rightarrow \lim_{t \to \infty} \dfrac{R(t)}{\mu}
= \lim_{t \to \infty}  \mu \int_{0}^{t} I(s) ds.
\end{equation*}
Therefore, $\lim_{t \to \infty} \int_{0}^{t} I(s) ds$ converge.
Thus $\sum_{n=0}^{\infty}I(n)$ is convergent and
$\lim_{t \to \infty}I(t)  =0$ \cite{bloch}.
\end{proof}

We have shown that $0 \leq S(t), I(t), R(t) \leq N$ over the entire time course
and $S(t)$, $R(t)$ and $I(t)$ converge to finite numbers as $t \to \infty$.
These results are important in that they allow us to examine the model even further.
The important point we want to address now is: how many secondary infectious
appear in the population after we introduce one infective into it?
If everyone is initially susceptible ($S_0 \approx N$), then
\begin{equation*}
\dfrac{dI}{dt} =
\beta S I - \mu I
= \left( \beta S  - \mu \right)I
\approx \left( \beta N  - \mu \right)I.
\end{equation*}
It means a newly introduced infected individual can be expected to infect other
people at the rate $\beta N$ during the expected infectious period $\dfrac{1}{\mu}$.
Thus, this first infective individual can be expected to infect
$R_0=\dfrac{\beta N}{\mu}=\dfrac{N}{p}$ individuals.
The number $R_0$ is called the basic reproduction number, which is one of the
most important quantities to be considered when analyzing any epidemic model.
It is a decreasing function for $R_0<1$, in which case the population
of infectious dies out naturally. However, if $R_0>1$, then it is an increasing
function and the virus spreads. More precisely, the virus spreads if $N>1$,
respectively $S_0>p$, i.e., the initial number of susceptible exceeds
the threshold value $p$; otherwise the virus will die out.
To sum it up: if a small group of infectious is inserted into a group of
susceptible $S_0$ and $S_0<p$, then the virus will die out rapidly.
On the other hand, if $S_0>p$, then $I(t)$ increases as $S(t)$ decreases to $p$,
where it achieves its maximum value, then $I(t)$ starts to decrease (when $S(t)$
falls bellow the threshold value $p$). These conclusions are summarized in the
following result.

\begin{theorem}
\label{theorm1:SIR}
Let $(S(t),I(t),R(t))$ be a solution of the SIR model
\eqref{SIRmodel}--\eqref{eq4:SIR} in
$T = \left\{ (S,I)  :  S \geq 0,  I \geq 0,   S+I \leq N \right\}$.
If $S_0<p$, then $I(t)$ decreases to $0$ as $t \to \infty$.
If $S_0>p$, then $I(t)$ increases until it attains its maximum value.
After that point, and according to Lemma~\ref{lemm6:SIR},
$I(t)$ decreases  to $0$ as $t \to \infty$; also,
$S(t)$ is a decreasing function and the limiting value $S_\infty$ is
the unique root of the equation
$I_0+S_0-S_\infty + p\ln\left(\dfrac{S_\infty}{S_0}\right)=0$.
\end{theorem}


\subsection{Analysis of the SEIR model}

We also prove some fundamental properties
of the SEIR system \eqref{eq1:SEIR}--\eqref{eq5:SEIR}.

\begin{proposition}
\label{prop:N:SEIR:const}
The total population $\hat{N}$ being modelled by the SEIR model \eqref{eq1:SEIR}
is constant in time, that is, $\hat{S}(t) + \hat{E}(t) + \hat{I}(t) + \hat{R}(t)
=  \hat{N}$ for all $t \geq 0$.
\end{proposition}

\begin{proof}
Follows by adding the right-hand side of equations
\eqref{eq1:SEIR}: $\hat{S}'(t)+\hat{E}'(t)+\hat{I}'(t)+\hat{R}'(t)=0$.
\end{proof}

\begin{lemma}
\label{lemma1:SEIR}
Let $\left(S(t), E(t), I(t), R(t)\right)$ be a solution
of the SEIR model \eqref{eq1:SEIR}--\eqref{eq5:SEIR}.
Then  $S(t) \geq 0$, $E(t) \geq 0$ and $I(t) \geq 0$, for all $t \geq 0$.
\end{lemma}

\begin{proof}
The proof is done by contradiction.
We have that $S(t)$ is a continuous function and $S(0)=\hat{S}_0>0$
is the initial condition. Suppose that $S(t')<0$. Then, by applying
the intermediate value theorem, there exist $t_1 \in [0,t']$ such that $S(t_1)=0$.
By using the first equation of \eqref{eq1:SEIR}, we obtain
$S(t)=S(t_1)e^{f}=0$ for $t \geq t_1$, where $f$ is the primitive of
$-\hat{\beta} I(t)$. So $S(t)=0$ for $t \geq t_1$, which is a contradiction.
Then $S(t)\geq 0$. Similarly, we prove that $E(t)\geq 0$ and $I(t)\geq 0$.
\end{proof}

\begin{lemma}
\label{lemma2:SEIR}
Let $\left(S(t), E(t), I(t), R(t)\right)$ be a solution
of the SEIR model \eqref{eq1:SEIR}--\eqref{eq5:SEIR}.
Then $0 \leq R(t) \leq \hat{N}$, $S(t) \leq \hat{N}$, $I(t) \leq \hat{N}$,
and $E(t) \leq \hat{N}$ for all $t\geq 0$.
\end{lemma}

\begin{proof}
From the fourth equation of \eqref{eq1:SEIR}, we have
$\dfrac{dR(t)}{dt} = \hat{\mu} I(t)$. Then $\dfrac{dR(t)}{dt} \geq 0$ for all
$t\geq 0$ and $R(t) \geq \hat{R}_0 =0$. By Lemma~\ref{lemma2:SEIR}, we know that
$S(t) \geq 0$, $E(t) \geq 0$ and $I(t) \geq 0$, so $-S(t)-E(t)-I(t) \leq 0$.
Therefore, $R(t)=\hat{N}-S(t)-E(t)-I(t) \leq \hat{N}$, and
$R(t)\leq \hat{N}$ for  all  $t\geq 0$. For the susceptible individuals,
$S(t)=\hat{N} -E(t)-I(t)-R(t) \leq \hat{N}$ because $E(t)+I(t)+R(t) \geq 0$.
Similarly, $E(t) \leq \hat{N}$ and $I(t) \leq \hat{N}$.
\end{proof}

\begin{lemma}
\label{lemm3:SEIR}
Let $\left(S(t), E(t), I(t), R(t)\right)$ be a solution
of the SEIR model \eqref{eq1:SEIR}--\eqref{eq5:SEIR}.
Then $R(t) \to \hat{R}_\infty$ as $t \to \infty$, where
$\hat{R}_\infty$ is a finite number.
\end{lemma}

\begin{proof}
From the fourth equation of \eqref{eq1:SEIR}, we have
$\frac{dR}{dt} = \hat{\mu} I(t) \geq 0$,
so that $R(t)$ is an increasing function. By Lemma~\ref{lemma2:SEIR},
$R(t)$ is bounded. Therefore, $\lim_{t\to \infty}R(t) = \hat{R}_\infty$,
where $\hat{R}_\infty$ is a finite number.
\end{proof}

\begin{lemma}
\label{lemm4:SEIR}
Let $\left(S(t), E(t), I(t), R(t)\right)$ be a solution
of the SEIR model \eqref{eq1:SEIR}--\eqref{eq5:SEIR}.
Then $I(t) \to 0$ as $t \to \infty$.
\end{lemma}

\begin{proof}
By using \eqref{eq1:SEIR}, and knowing from Lemma~\ref{lemm3:SEIR} that
$\lim_{t \to \infty} R(t) = \hat{R}_\infty$, so that
$\lim_{t \to \infty} \hat{\mu} \int_0^t I(s) ds
= \dfrac{\hat{R}_\infty}{\hat{\mu}}$,
we obtain
\begin{equation*}
\int\dfrac{dR}{dt} = \int \hat{\mu} I dt
\Rightarrow  R(t) = \hat{\mu} \int_{0}^{t} I(s) ds
\Rightarrow \lim_{t \to \infty} \dfrac{R(t)}{\hat{\mu}}
= \lim_{t \to \infty}  \hat{\mu} \int_{0}^{t} I(s) ds.
\end{equation*}
Therefore, $\lim_{t \to \infty} \int_{0}^{t} I(s) ds$ converges,
so $\sum_{n=0}^{\infty}I(n)$ is convergent \cite{bloch}.
We conclude that $I(t) \to 0$ as  $t \to \infty$.
\end{proof}


\section{Numerical simulation of the models}
\label{Sec:3}

Now we investigate the numerical simulation of the SIR and SEIR models
proposed for the description of the transmission of the Ebola virus
in the recent outbreak occurred in West Africa \cite{MyID:321,symcomp}.
In the numerical resolution of the nonlinear differential equations
of both SIR and SEIR models, we use the \textsf{Matlab} differential equation
solver \texttt{ode45}. This routine uses a variable step Runge--Kutta
4th/5th-order method to solve numerically the differential equations.

Rachah and Torres considered the system of equations that describe
the SIR model by using the transmission rate parameter $\beta=0.2$,
the recovered rate $\mu=0.1$ and the values
$\left(S(0),I(0),R(0)\right) =\left(0.95, 0.05, 0\right)$
for the initial number of susceptible, infected, and recovered individuals.
These parameters of the SIR model were identified by using recent data from
the World Health Organisation (WHO) \cite{MyID:321}.
Also motivated by the data of WHO, in their numerical study
of the propagation of the virus by the SEIR model,
they considered the SEIR model with a transmission rate $\hat{\beta}=0.2$,
an infectious rate $\hat{\gamma}=0.1887$, a recovered rate $\hat{\mu}=0.1$,
and initial values $\left(\hat{S}(0),\hat{E}(0),\hat{I}(0),\hat{R}(0)\right)
=\left(0.88, 0.07, 0.05, 0\right)$ for the initial number of susceptible,
exposed, infected, and recovered individuals \cite{symcomp}. The numerical
simulations of the SIR and SEIR models were carried out in the previous works
of Rachah and Torres \cite{MyID:321,symcomp}. Here, by using the above mentioned
parameters and initialization values, we investigate the possibility
to obtain curves of infectious that converge at the same time
and have the same maximum numbers of infectious individuals
by the two models.

After an extensive set of numerical tests of the models with several values
of $\hat{\gamma}$, we concluded that convergence between the results
of the two models is possible by increasing the infectious rate. In fact, when
$\hat{\gamma}$ increases, the individuals enter the $I$ class as soon
as the individuals enter the class $E$. In that case, the SIR model
can be viewed as a special case of the SEIR model.
In the analysis and comparison between the numerical resolution of the SIR
and SEIR models, we used the real parameters identified from the real data of the WHO
given by $\hat{\beta}=\beta=0.2$, $\hat{\mu}=\mu=0.1$,
the initial values $\left(\hat{S}(0),\hat{E}(0),\hat{I}(0),\hat{R}(0)\right)
=\left(0.88, 0.07, 0.05, 0\right)$, and several coefficients values
of $\hat{\gamma}$. Our comparison between the numerical results
of the SIR and SEIR models is shown in Figure~\ref{sml_modls_2p47gam}.
\begin{figure}
\centering
\subfloat[Infectious individuals]{%
\includegraphics[scale=0.54]{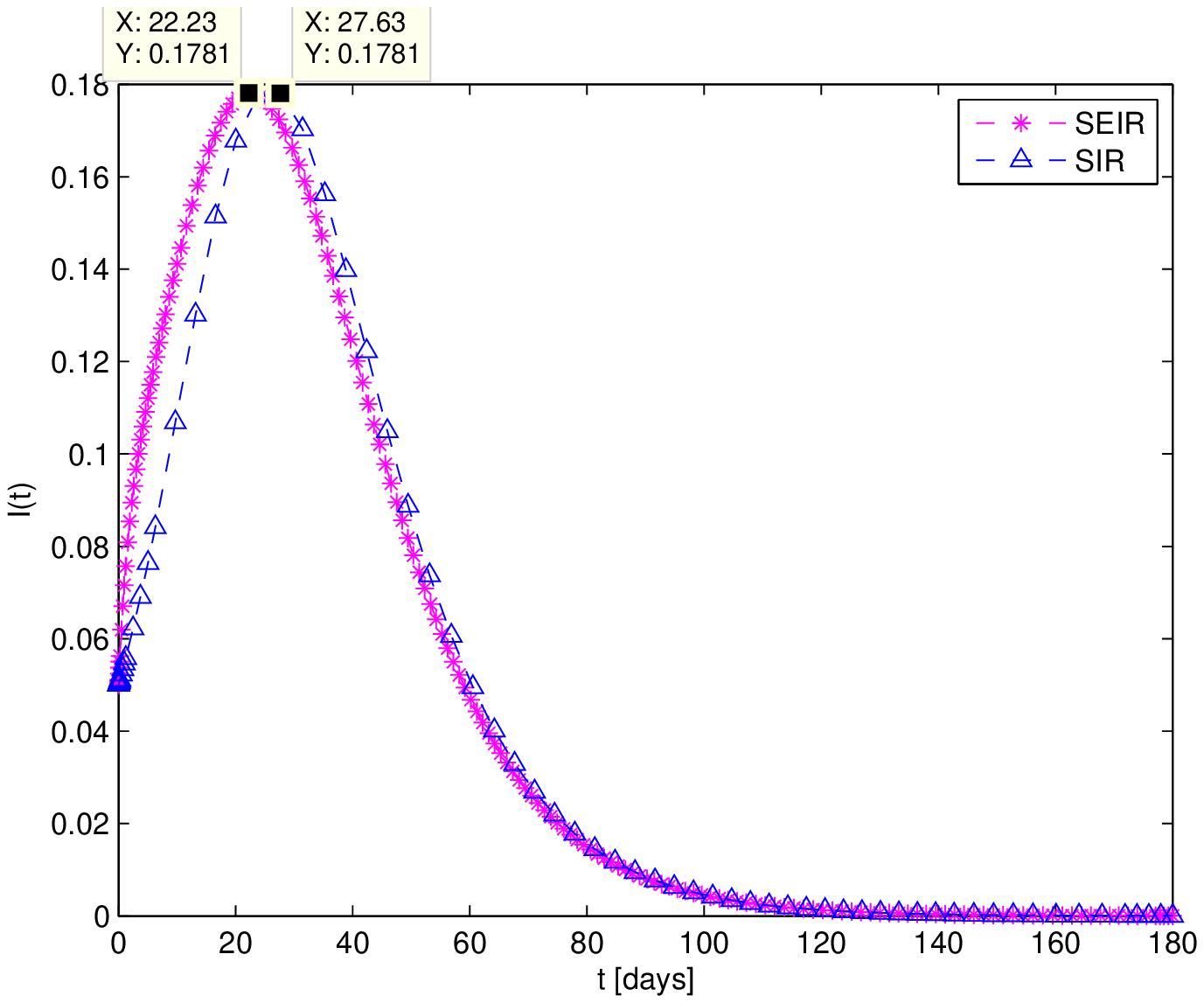}}
\subfloat[Recovered individuals]{%
\includegraphics[scale=0.54]{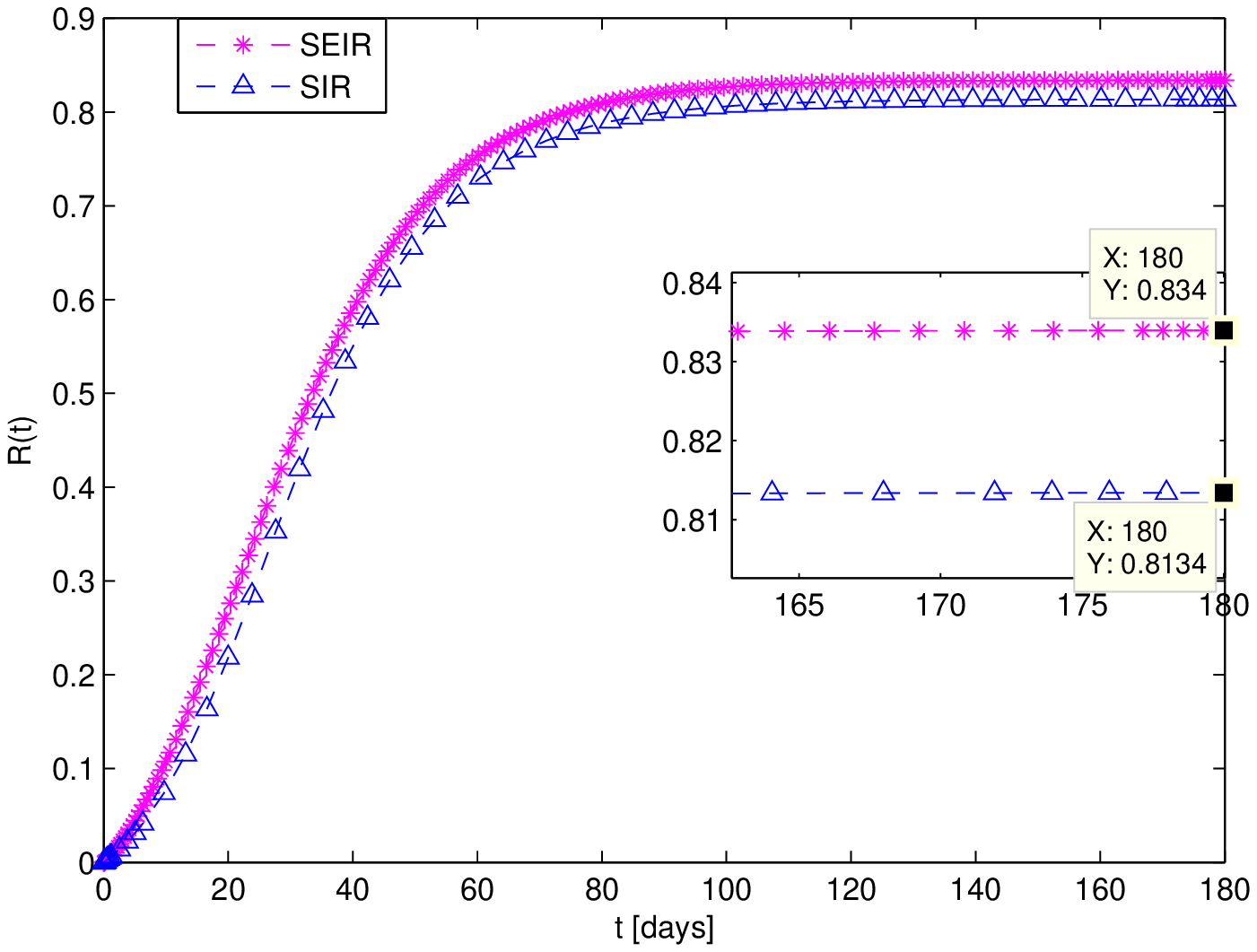}}		
\caption{Comparison between the numerical solution of the SIR and SEIR models
in the case of $\hat{\gamma}= 2.47\times 0.1887$.
\label{sml_modls_2p47gam}}
\end{figure}
The numerical study of the models shows that they have the same maximum number
and a convergence of the curve of infectious at the same time for the value
of infectious rate $\hat{\gamma}= 2.47\times 0.1887$. The meaning of this value
is that when $\hat{\gamma}$ increases, the number of exposed individuals decreases.
Then we obtain a convergence between the curves of infectious of the two models.

For the SEIR model, we chose $\hat{\gamma}= 2.47\times 0.1887$
with a small value of $E(0)$, $\hat{\beta}=\beta$, $\hat{\mu}=\mu$
and $\left(\hat{I}(0),\hat{R}(0)\right) =\left(I(0),R(0)\right)$,
in order to discuss the difference between the models in case
when they have the same maximum number of infectious
with similar curves of infectious. Figure~\ref{sml_modls_2p47gam}
shows the curve of infectious individuals of the SIR model,
which is identical to the SEIR curve of infectious individuals.
The same figure shows that the maximum number of infectious individuals
is the same for the two models and the curves converge at the same time,
and presents the number of recovered individuals for the two models
and for the same maximum number of infectious. By using the SEIR model,
the recovered number of individuals is more important than the number
of recovered individuals of the SIR model. The numerical results show that
even when one model converge to the other one (with respect to the infectious number),
by increasing $\hat{\gamma}$ the SEIR model describes better the propagation
of the virus, which is closer to the reality of Ebola virus, characterized
by a period of incubation described by the exposed group. Even when the number
of exposed is very small, it provides a more detailed description
of the Ebola virus.


\section{Optimal control of the virus by using SIR and SEIR models}
\label{Sec:4}

We now address the question of how to optimally control
the propagation of the spread of Ebola in a population
by using the SIR and SEIR models. Let us start by the
optimal control strategy based on the SIR model,
as studied by Rachah and Torres in \cite{MyID:321},
which is given by the following system
of nonlinear differential equations:
\begin{equation}
\label{SIR_control}
\begin{cases}
\dfrac{dS(t)}{dt} = -\beta S(t)I(t) - u(t) S(t),\\[0.3cm]
\dfrac{dI(t)}{dt} = \beta S(t)I(t) - \mu I(t),\\[0.3cm]
\dfrac{dR(t)}{dt} = \mu I(t) + u(t) S(t).
\end{cases}
\end{equation}
Note that if $u(t) \equiv 0$, then \eqref{SIR_control} reduces to \eqref{SIRmodel}.
Although there are currently no licensed Ebola vaccines, two potential candidates
are undergoing evaluation. The goal of the strategy is to reduce the infected
individuals and the cost of vaccination. Precisely, the optimal control problem
consists of minimizing the objective functional
\begin{equation}
\label{cost_func_strat_sir}
J(u) = \int_{0}^{t_{end}} \left[I(t) + \dfrac{\nu}{2}u^2(t)\right] dt,
\end{equation}
where $u(t)$ is the control variable, which represents the vaccination rate
at time $t$, and the parameters $\nu$ and $t_{end}$ denote, respectively,
the weight on cost and the duration of the vaccination program.

Let us now study the same cost functional of the previous optimal control problem,
by introducing into the model \eqref{eq1:SEIR} a control $\tilde{u}(t)$,
representing the vaccination rate at time $t$. The control $\tilde{u}(t)$
is the fraction of susceptible individuals being vaccinated per unit of time.
Then, the mathematical model with control is given by the following
system of nonlinear differential equations:
\begin{equation}
\label{SEIR_control}
\begin{cases}
\dfrac{d\tilde{S}(t)}{dt}
= -\tilde{\beta} \tilde{S}(t)\tilde{I}(t) - \tilde{u}(t) \tilde{S}(t),\\[0.30cm]
\dfrac{d\tilde{E}(t)}{dt}
= \tilde{\beta} \tilde{S}(t)\tilde{I}(t) - \tilde{\gamma} \tilde{E}(t),\\[0.30cm]
\dfrac{d\tilde{I}(t)}{dt}
= \tilde{\gamma} \tilde{E}(t) - \tilde{\mu} \tilde{I}(t),\\[0.30cm]
\dfrac{d\tilde{R}(t)}{dt} = \tilde{\mu} \tilde{I}(t) + \tilde{u}(t) \tilde{S}(t).
\end{cases}
\end{equation}
The goal of our strategy is to reduce the infected individuals and the cost
of vaccination. Precisely, the optimal control problem consists of minimizing
the objective functional
\begin{equation}
\label{cost_func_strat_seir}
J(\tilde{u}) = \int_{0}^{t_{end}} \left[\tilde{I}(t)
+ \dfrac{\tilde{\tau}}{2}\tilde{u}^2(t)\right] dt,
\end{equation}
where $\tilde{u}(t)$ is the control variable, which represents the vaccination
rate at time $t$, and the parameters $\tilde{\tau}$ and $t_{end}$ denote,
respectively, the weight on cost and the duration of the vaccination program.

Before comparing between the results of control of the two models,
we present in Figure~\ref{sml_cntrl_SIR} the numerical simulation
of the SIR model without control compared with the study of control
of the virus described by the system \eqref{SIR_control}
and the cost functional \eqref{cost_func_strat_sir}. Figure~\ref{sml_cntrl_SIR}
shows the effect of the optimal control strategy in reducing the number of
infectious and the period of infection, and in increasing the number of recovered.

Figure~\ref{sml_cntrl_SEIR} shows the numerical simulation of the SEIR model
without control compared with the study of control of the virus described
by the system \eqref{SEIR_control} and the cost functional
\eqref{cost_func_strat_seir}. The curves of the compartments
in Figure~\ref{sml_cntrl_SEIR} show the effect of the optimal control strategy
in reducing the number of exposed and infectious individuals and the period of
infection, and in increasing the number of recovered individuals.
\begin{figure}
\centering
\subfloat[Susceptible with and without control]{%
\includegraphics[scale=0.54]{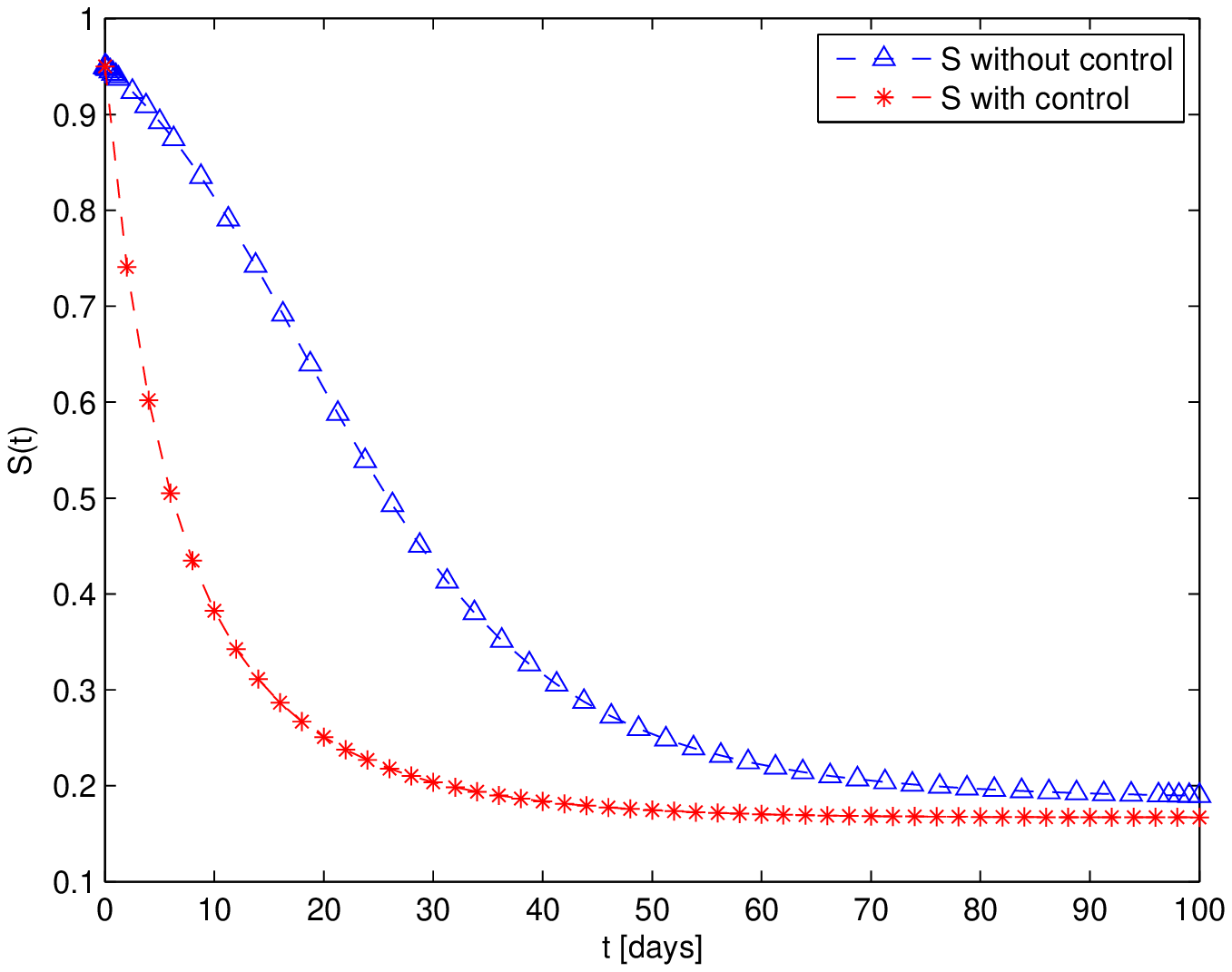}} 	
\subfloat[Infectious with and without control]{%
\includegraphics[scale=0.54]{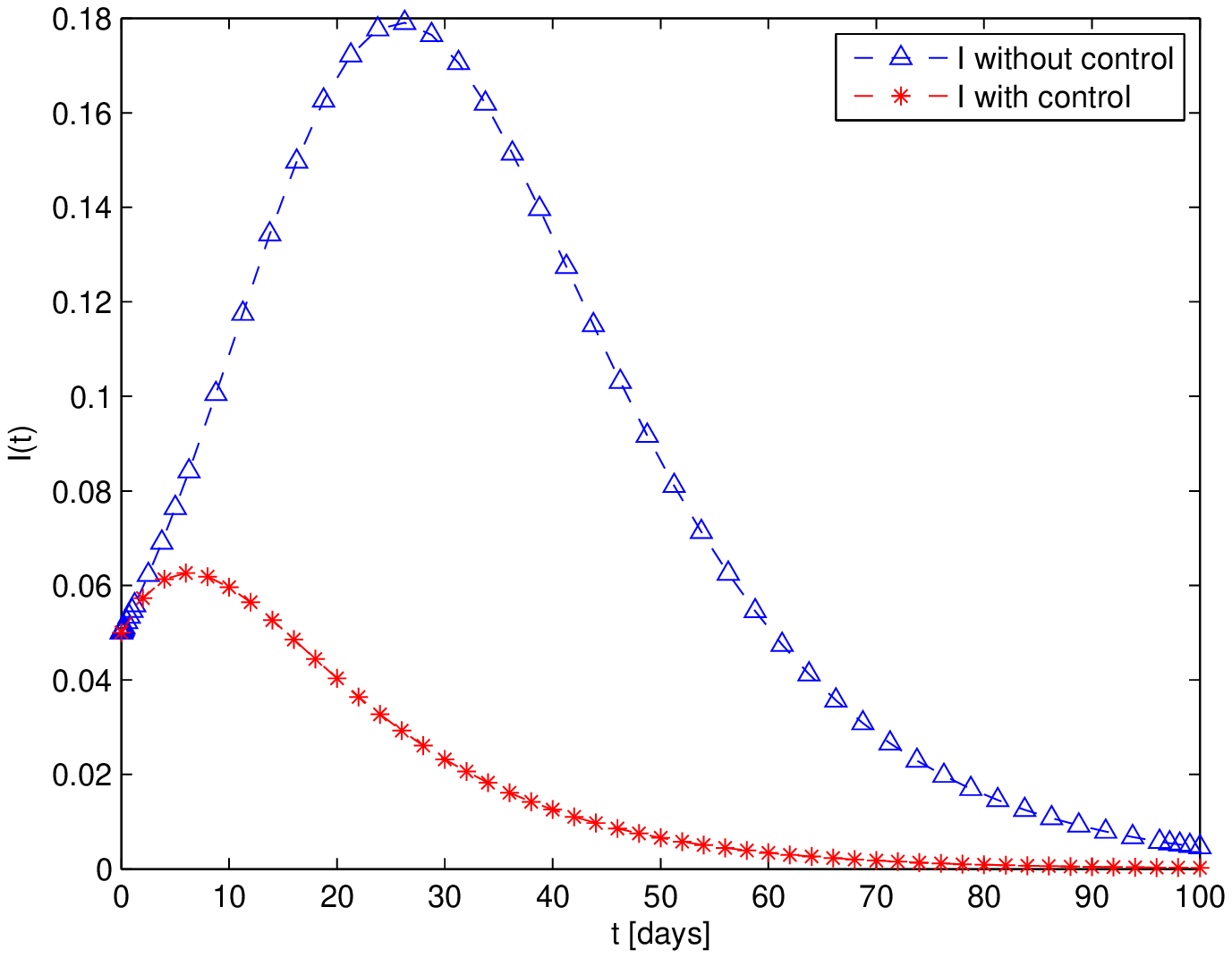}}\\
\subfloat[Recovered with and without control]{%
\includegraphics[scale=0.54]{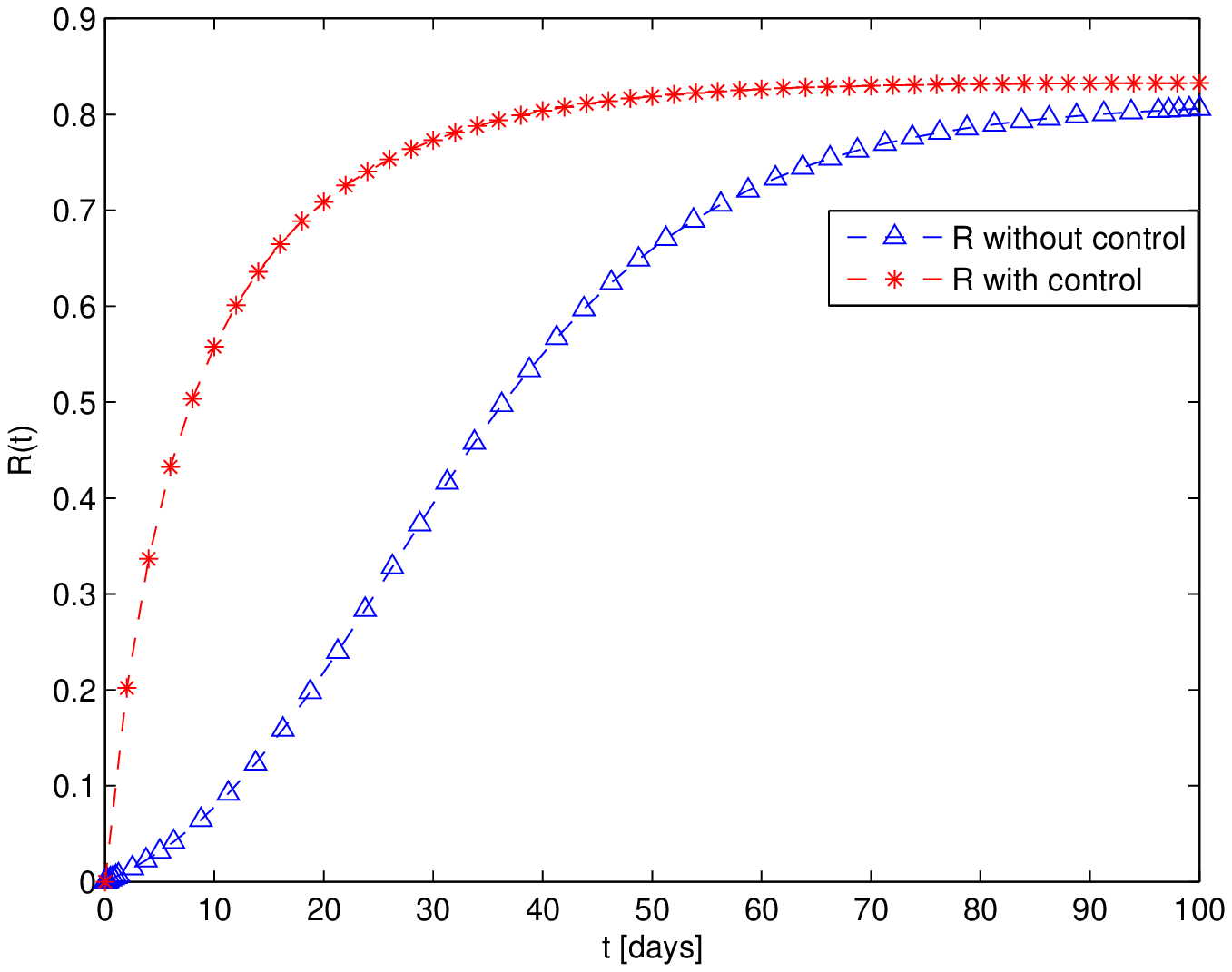}}
\caption{Simulation of Ebola virus and numerical resolution of a strategy
control of the Ebola virus by using the SIR control model described by system
\eqref{SIR_control} and the cost functional \eqref{cost_func_strat_sir}.
\label{sml_cntrl_SIR}}
\end{figure}
\begin{figure}
\centering
\subfloat[Susceptible with and without control]{%
\includegraphics[scale=0.54]{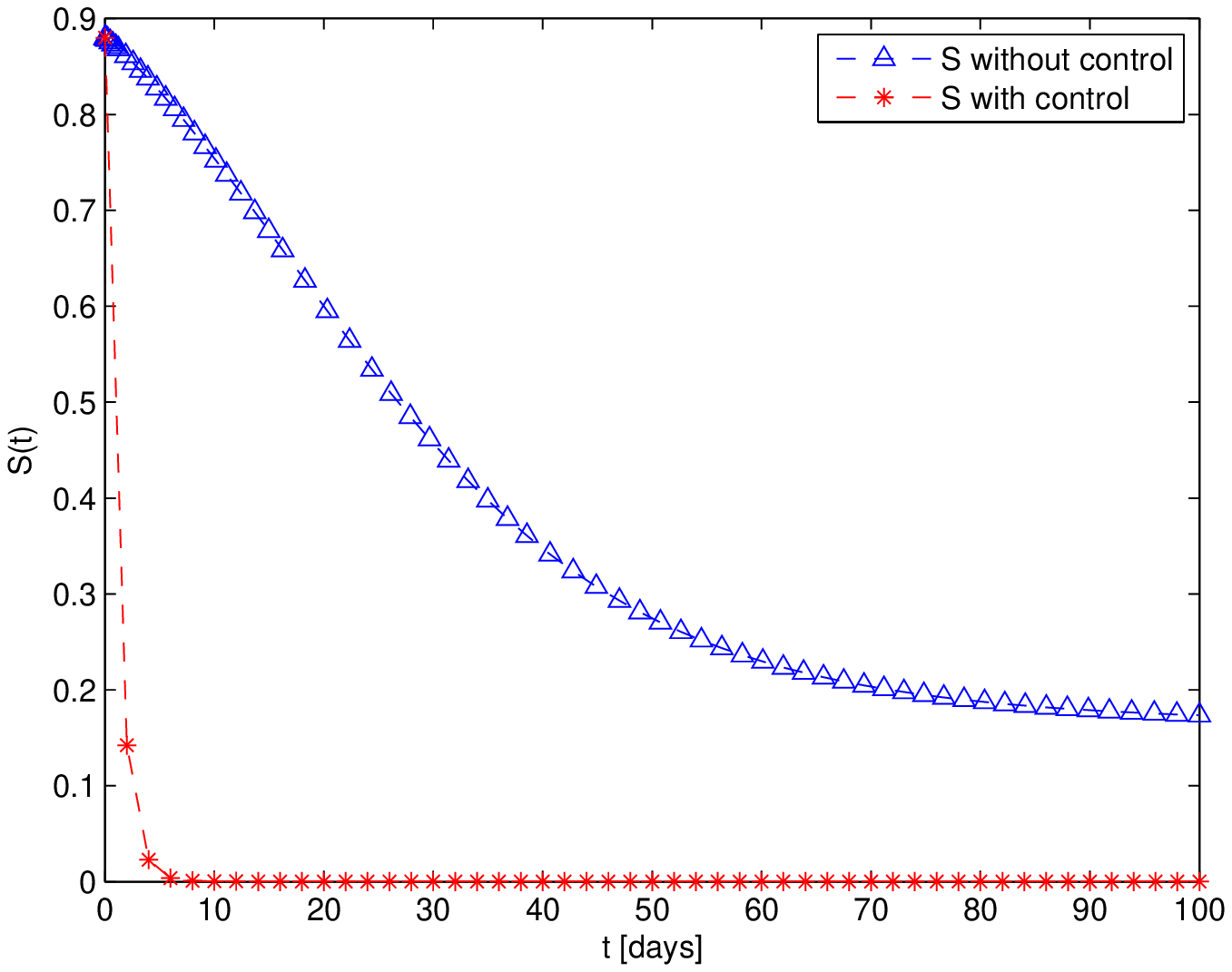}}
\subfloat[Exposed with and without control]{%
\includegraphics[scale=0.54]{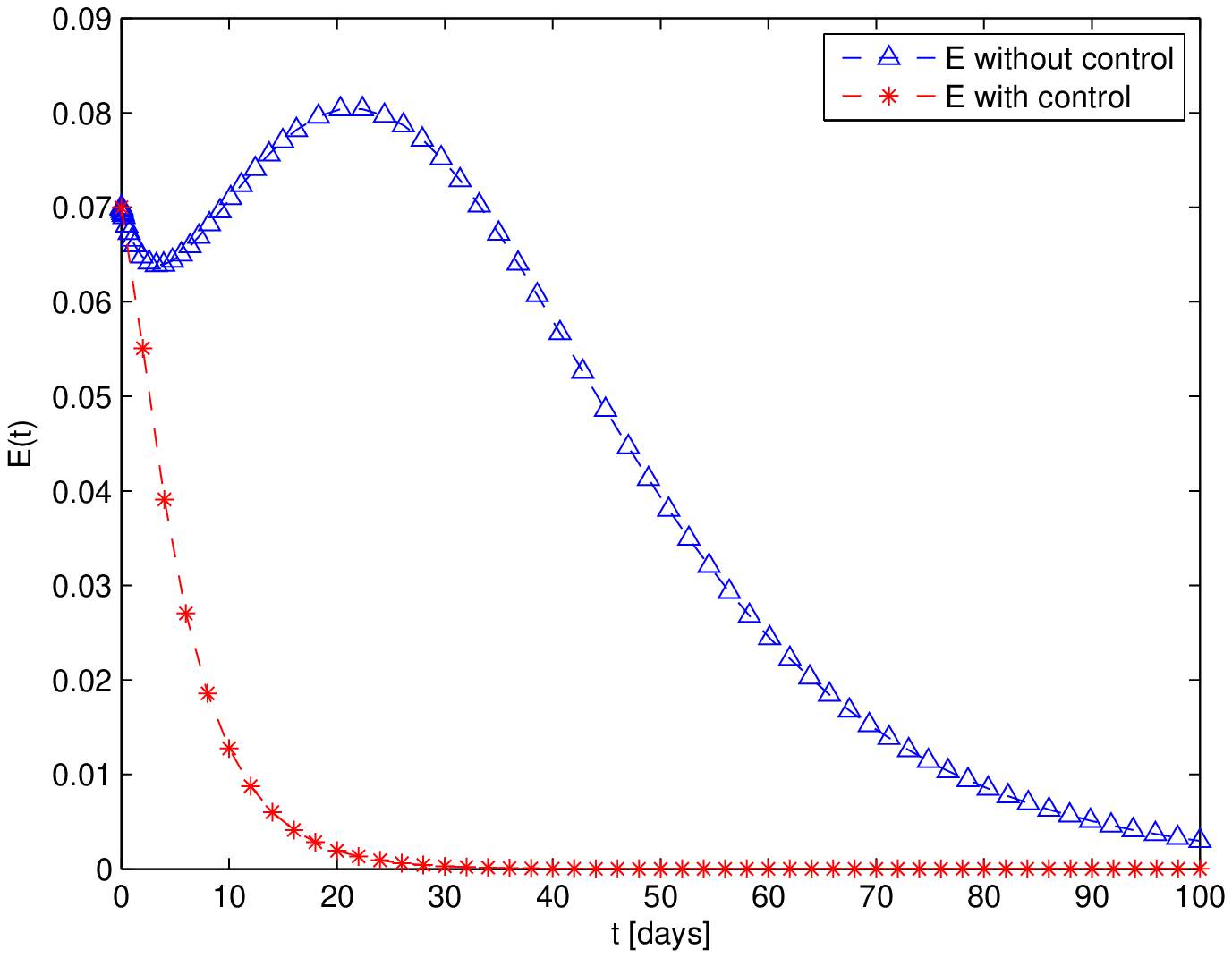}}\\	
\subfloat[Infectious with and without control]{%
\includegraphics[scale=0.54]{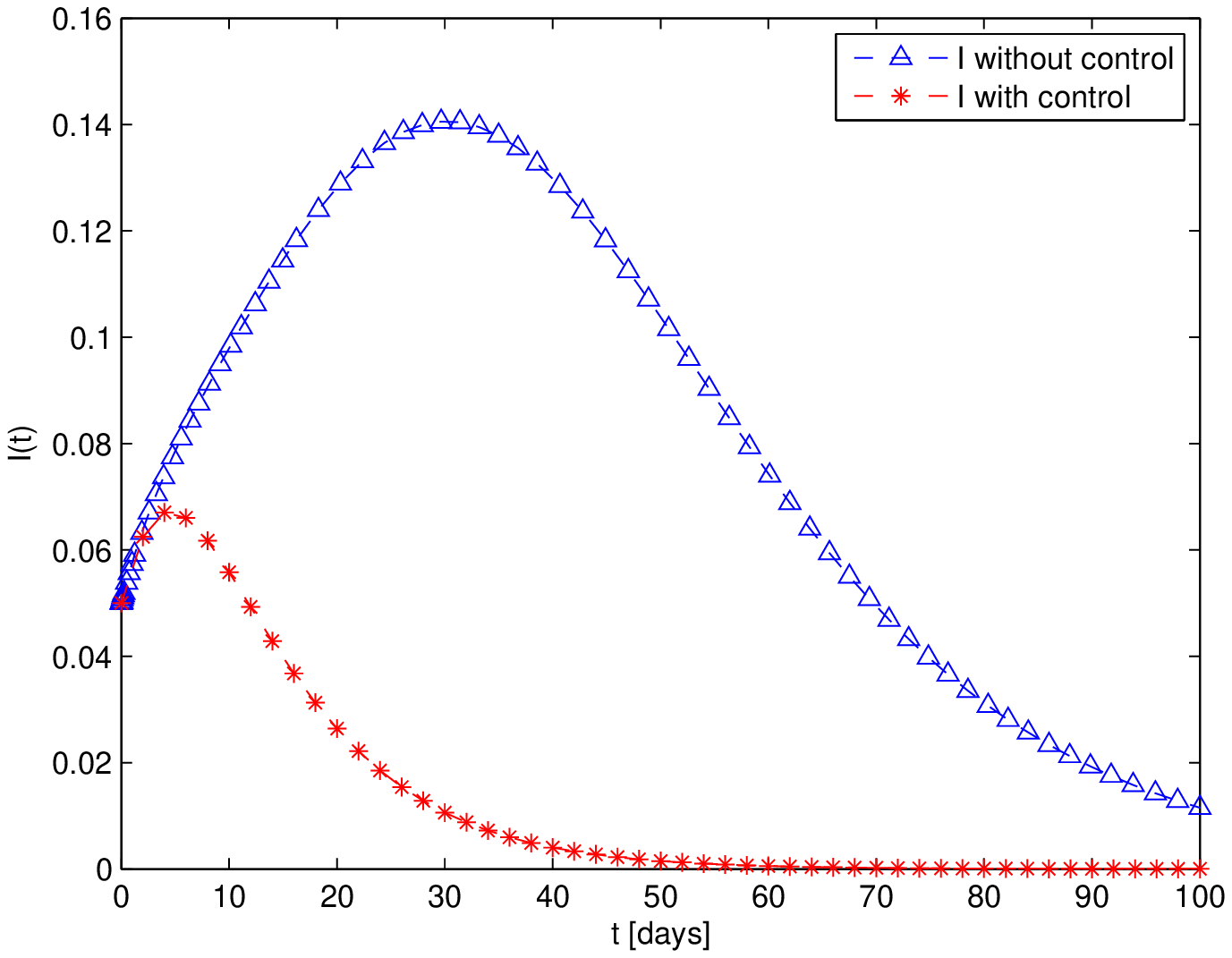}}
\subfloat[Recovered with and without control]{%
\includegraphics[scale=0.54]{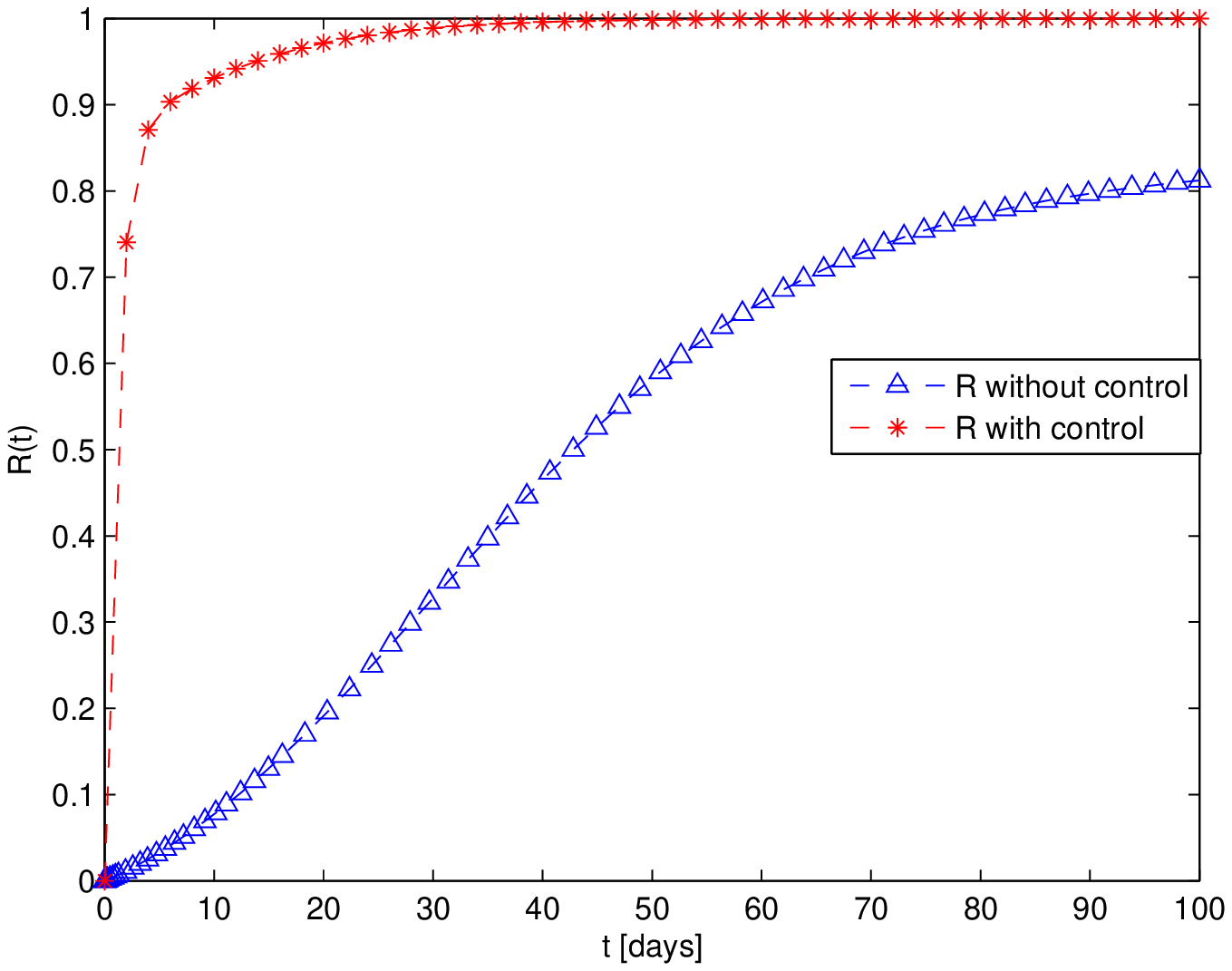}}	
\caption{Simulation of Ebola virus and numerical resolution of a strategy
control of the Ebola virus by using the SEIR control model described by system
\eqref{SEIR_control} and the cost functional \eqref{cost_func_strat_seir}.
\label{sml_cntrl_SEIR}}
\end{figure}

Let us now discuss the comparison between the results of the control strategy
of the SIR and SEIR models. Figure~\ref{compr_cntrl_SEIR_SIR_S} shows that the
number of susceptible decreases faster in case of control of the virus
by using the SEIR model than the case of control by the SIR model.
Figure~\ref{compr_cntrl_SEIR_SIR_I} presents the number of
infectious individuals, which is characterized by a peak in the same time
(4 days) for the two models with a number of infected given by 6.7\%
in case of control by the SEIR model and 5.6\%  in the case of control
by the SIR model. The difference in the number of infectious between
the two models is explained by the fact that the exposed
individuals certainly moved to the infectious class since
the Ebola virus is characterized by a short latent period,
then it spreads very quickly.
The effect of the control strategy is shown in the period of infection,
which is given by 66 days in case of control
by the SEIR model, which is less than the infectious period (78 days) in case
of control by using the SIR model. The number of recovered individuals
is shown in Figure~\ref{compr_cntrl_SEIR_SIR_R}, where it reaches 99\%
at the end of the campaign in case of control by the SEIR model versus 88\%
at the end of the campaign in case of control with the SIR model.
\begin{figure}
\centering
\subfloat[Susceptible]{\includegraphics[scale=0.54]{%
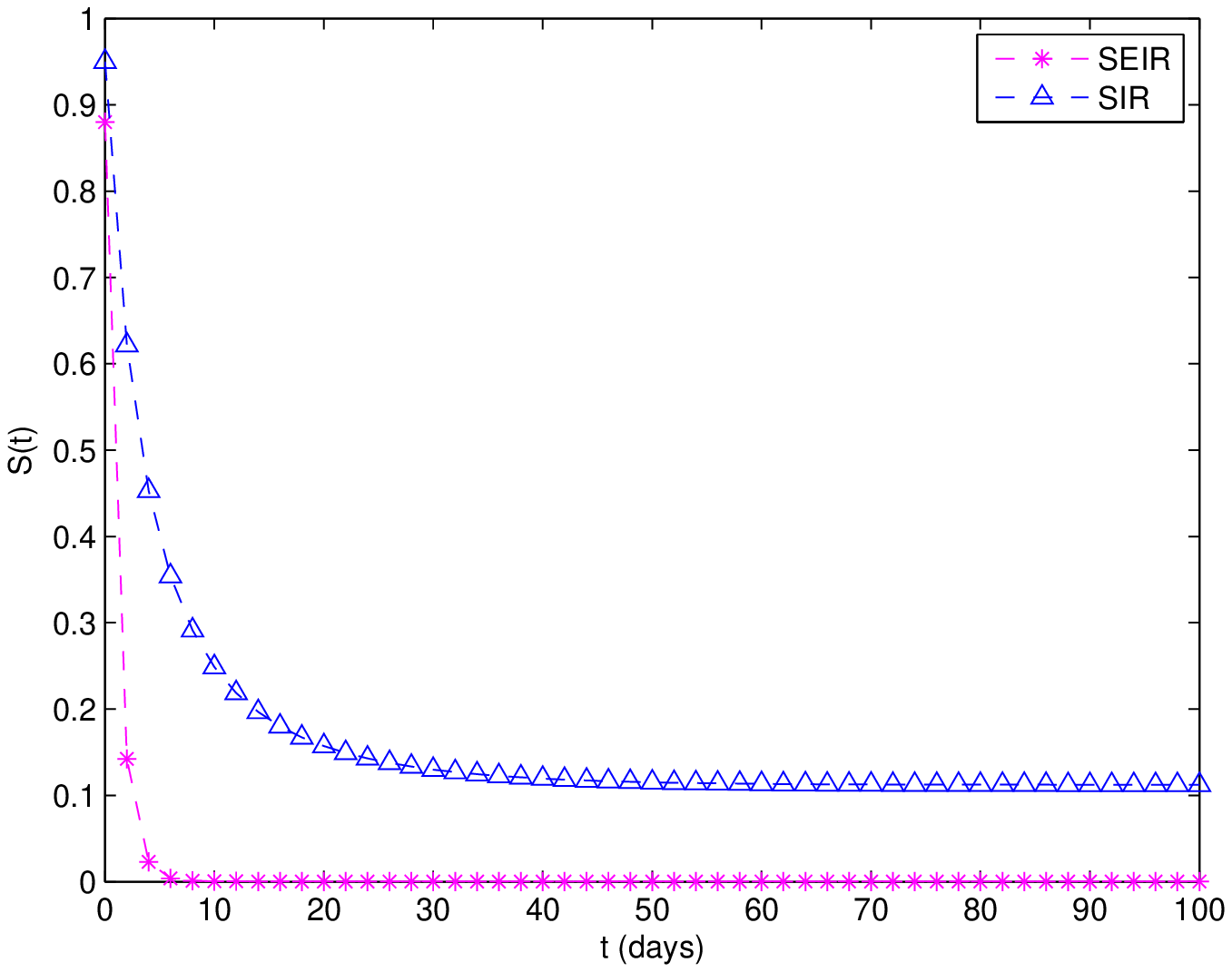}\label{compr_cntrl_SEIR_SIR_S}}
\subfloat[Infectious]{\includegraphics[scale=0.54]{%
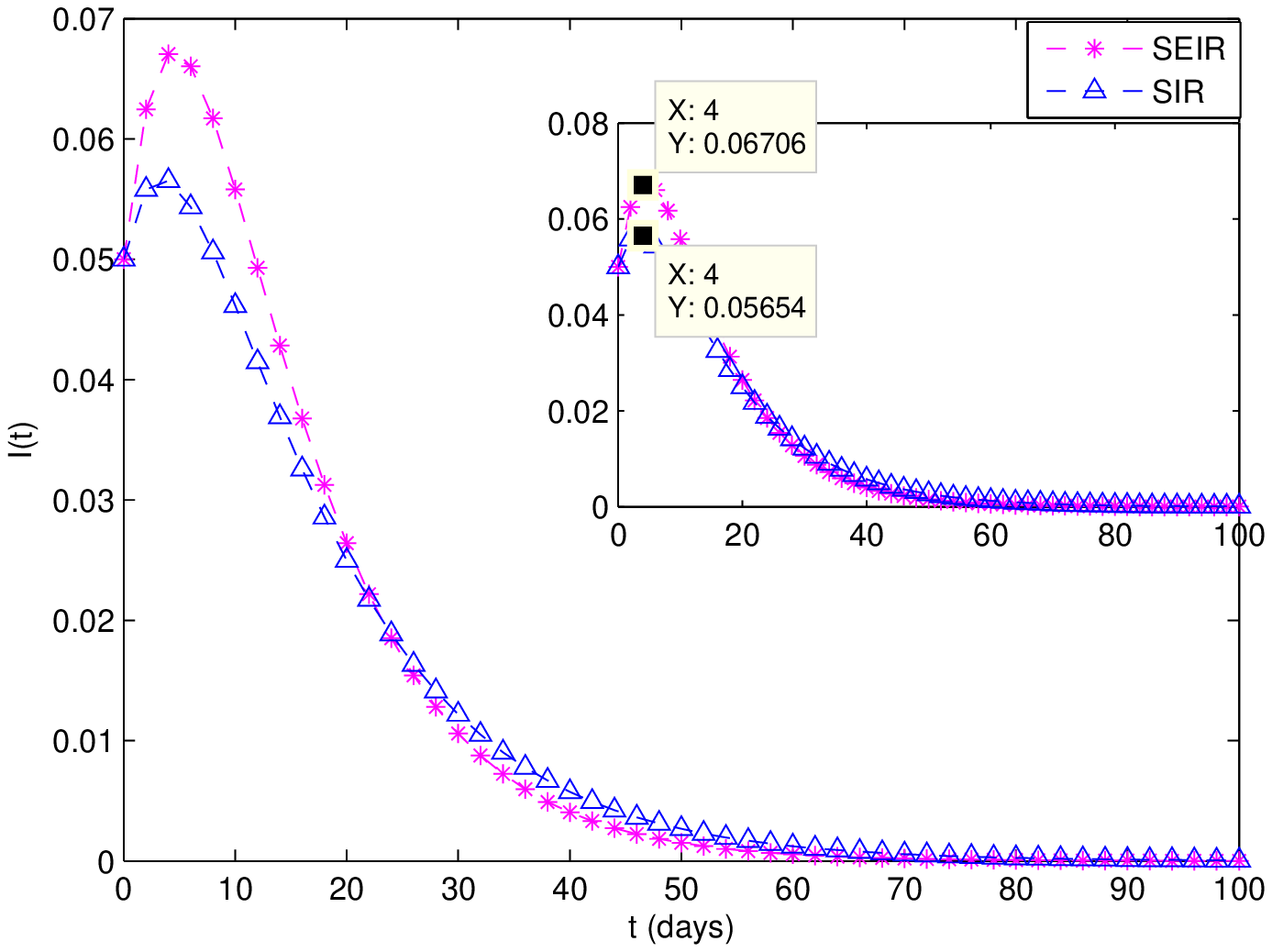}\label{compr_cntrl_SEIR_SIR_I}}\\
\subfloat[Recovered]{%
\includegraphics[scale=0.54]{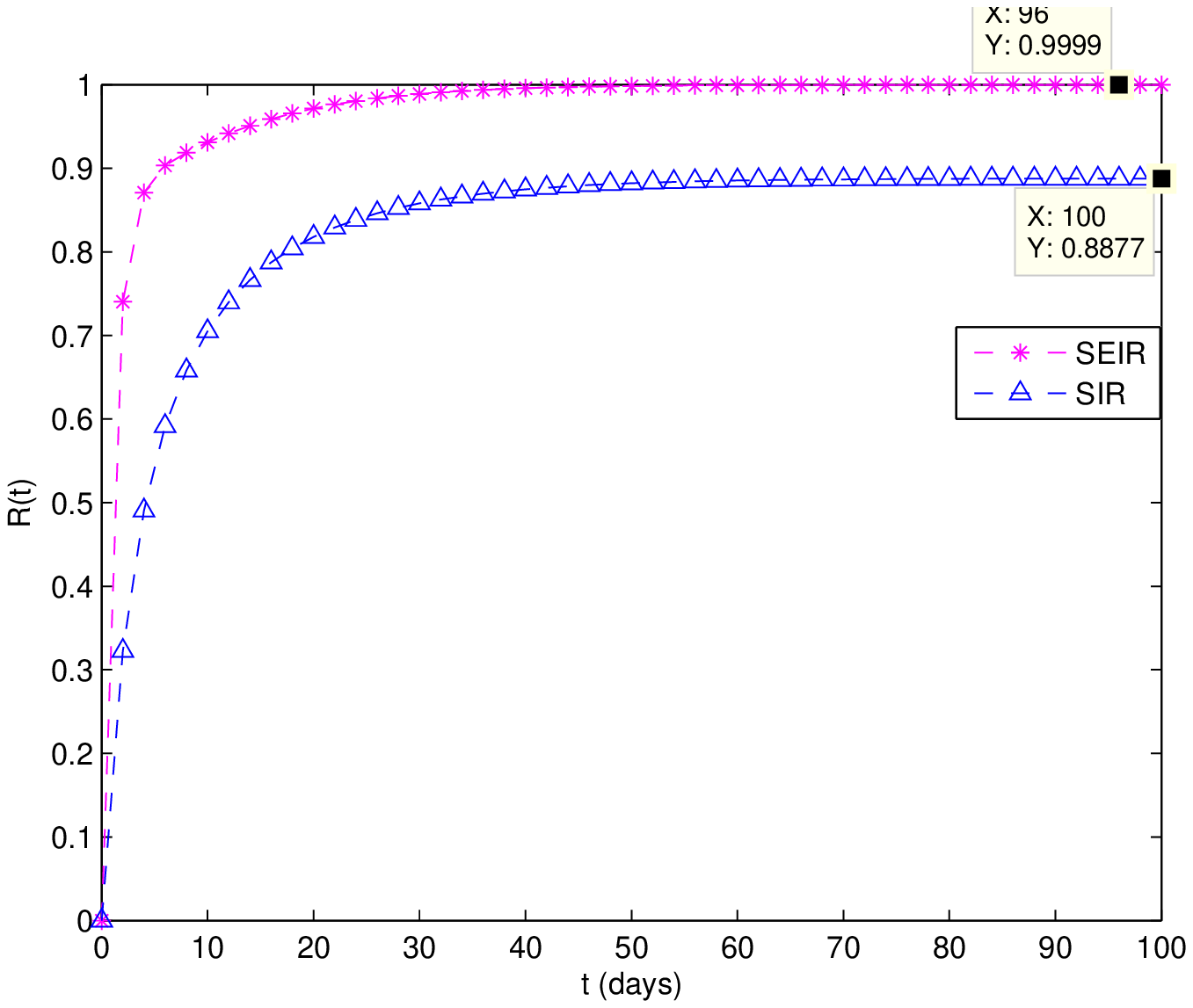}%
\label{compr_cntrl_SEIR_SIR_R}}
\subfloat[Optimal control]{%
\includegraphics[scale=0.54]{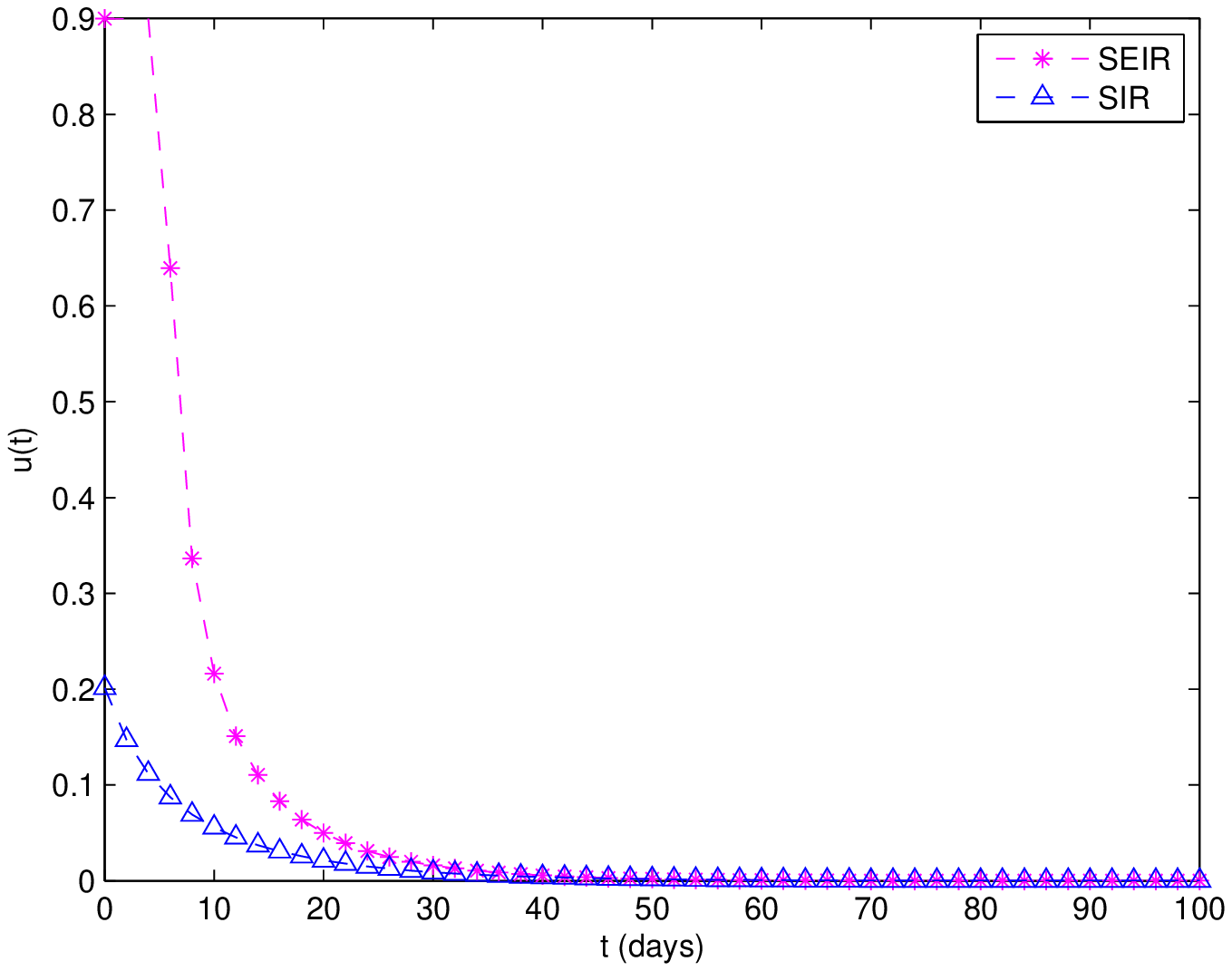}\label{cntrl_SEIR_seir}}
\caption{Comparison of control of Ebola virus by the SIR model
(described by system \eqref{SIR_control} and the cost functional
\eqref{cost_func_strat_sir}) versus control by the SEIR model
(described by system \eqref{SEIR_control} and the cost functional
\eqref{cost_func_strat_seir}).}
\end{figure}
Figure~\ref{cntrl_SEIR_seir} presents the control variable. One can see that
the control is more important in the SEIR model than in the SIR case.
More precisely, Figure~\ref{cntrl_SEIR_seir} shows that the optimal control function
of the SEIR model starts at the upper bound 0.9, while the optimal control
of the SIR model starts at 0.2. This fact is explained by the severity
of the spread of the virus, which requires an immediate implementation against
the virus, and the superiority of the SEIR system in modelling Ebola:
the exposed compartment of the SEIR model shows the severity of the virus,
where individuals can be infected in the beginning, without symptoms,
and then transmission blows up during the latent period. Thus, the rate
of vaccination is more crucial in case of the control strategy
based on the SEIR model, which is closer to the reality
of the propagation of the Ebola virus.


\section{Conclusion}
\label{Sec:conc}

We presented a comparison study between the SIR and SEIR models used in the
description of the propagation of the Ebola virus. The models were compared
by using their numerical simulation and also by studying optimal control
strategies for the control of the virus. Our investigations are based on the
parameters previously identified by Rachah and Torres in their study
of the Ebola virus \cite{MyID:321,symcomp}.


\ack{This research was supported by the
Institut de Math\'{e}matiques de Toulouse (Rachah);
and by the Portuguese Foundation for Science and Technology (FCT),
within R\&D unit CIDMA, project UID/MAT/04106/2013 (Torres).
The authors would like to thank two reviewers for their comments
and suggestions.}



\end{document}